\RequirePackage{fix-cm}
\documentclass[smallextended]{svjour3}       
\smartqed  
\usepackage{graphicx}
\usepackage{algorithm}
\usepackage{algorithmic}
\usepackage{tikz}
\usetikzlibrary{arrows,shapes, calc, fit, positioning}
\usepackage{amsmath,amssymb,amsfonts,ascmac}
\usepackage{booktabs}
\newcommand{\bbR}{\mathbb{R}}
\newcommand{\bbN}{\mathbb{N}}

\newcommand{\argmax}{\mbox{argmax}}

\usepackage[round]{natbib}

\newcommand{\calA}{\mathcal{A}}

\newcommand{\calU}{U}

\begin{document}

\title{Fair allocation of indivisible goods and chores
}

\titlerunning{Fair allocation of combinations of indivisible goods and chores}        

\author{Haris~Aziz \and Ioannis~Caragiannis \and Ayumi~Igarashi$^{*}$ \and Toby~Walsh
}

\institute{
H. Aziz \at
UNSW Sydney and Data61 CSIRO, Australia\\
\email{haris.aziz@unsw.edu.au}           
           \and
I. Caragiannis \at
Aarhus University, Denmark\\
\email{iannis@cs.au.dk}           
           \and
A. Igarashi$^{*}$ (corresponding author) \at
National Institute of Informatics, Japan\\
\email{ayumi\_igarashi@nii.ac.jp}           
           \and
T. Walsh \at
UNSW Sydney and Data61 CSIRO, Australia\\
\email{tw@cse.unsw.edu.au}           
}

\maketitle

\begin{abstract}
We consider the problem of fairly dividing a set of items. Much of the fair division literature assumes that the items are ``goods'' i.e., they yield positive utility for the agents. There is also some work where the items are ``chores'' that yield negative utility for the agents. In this paper, we consider a more general scenario where an agent may have positive or negative utility for each item. This framework captures, e.g., fair task assignment, where agents can have both positive and negative utilities for each task. We show that whereas some of the positive axiomatic and computational results extend to this more general setting, others do not. We present several new and efficient algorithms for finding fair allocations in this general setting. We also point out several gaps in the literature regarding the existence of allocations satisfying certain fairness and efficiency properties and further study the  complexity of computing such allocations.
\end{abstract}

\newpage
\section{Introduction}
Consider a group of students who are assigned to a certain set of coursework tasks. Students may have subjective views regarding how enjoyable each task is. For some people, solving a mathematical problem may be fulfilling and rewarding. For others, it may be nothing but torture. A student who gets more cumbersome chores may be compensated by giving her some valued goods so that she does not feel hard done by. 

This example can be viewed as an instance of a classic fair division problem. The agents have different preferences over the items and we want to allocate the items to agents as fairly as possible. The twist we consider is that whether an agent has positive or negative utility for an item is subjective. Our setting is general enough to encapsulate two well-studied settings: (1) ``good allocation'' in which agents have positive utilities for the items and (2) ``chore allocation'' in which agents have negative utilities for the items. The setting we consider also covers a third setting (3) ``allocation of objective goods and chores'' in which the items can be partitioned into chores (that yield negative utility for all agents) and goods (that yield positive utility for all agents). Setting (3) covers several scenarios where an agent could be compensated by some goods for doing some chores.

In this paper, we suggest a very simple yet general model of allocation of indivisible items that properly includes chore and good allocation. For this model, we present some case studies that highlight that whereas some existence and computational results can be extended to our general model, in other cases the combination of good and chore allocation poses interesting challenges not faced in subsettings. Our central technical contributions are several new efficient algorithms for finding fair allocations. In particular:
\begin{itemize}
    \item We formalize fairness concepts for the general setting. Some fairness concepts directly extend from the setting of good allocation to our setting. Other fairness concepts such as ``envy-freeness up to one item'' (EF1) and ``proportionality up to one item'' (PROP1) need to be generalized appropriately. 
    \item We show that the round robin sequential allocation algorithm that returns an EF1 allocation for the case of goods does not work in general. Nevertheless, we present a careful generalization of the decentralized round robin algorithm that finds an EF1 allocation when utilities are additive. 
    \item Turning our attention to an efficient and fair allocation, we show that for the case of two agents, there exists a polynomial-time algorithm that finds an EF1 and Pareto-optimal (PO) allocation for our setting. The algorithm can be viewed as an interesting generalization of the Adjusted Winner rule ~\citep{BT96a,BT96b} that is designed for divisible goods. 
    \item Weakening EF1 to PROP1, we show that there exists an allocation that is not only PROP1 but also contiguous (assuming that items are placed in a line). We further give a polynomial-time algorithm that finds such an allocation. 
\end{itemize}

\subsection{Related Work}
Fair allocation of indivisible items is a central problem in several fields including computer science and economics~\citep{BT96a,BCM15a}. Fair allocation has been extensively studied for allocation of divisible goods, commonly known as cake cutting~\citep{BT96a}. 

There are several established notions of fairness, including envy-freeness and proportionality. 
The recently introduced \emph{maximin share} (MMS) notion is weaker than envy-freeness and proportionality and has been heavily studied in the computer science literature. \citet{KPW18} showed that an MMS allocation of goods may not always exist; on the positive side, there exists a polynomial-time algorithm that returns a 2/3-approximate MMS allocation~\citep{KPW18,AMNS17}. Subsequent papers have presented simpler \citep{BK20} or even better \citep{SGHSY18} approximation algorithms for MMS allocations. In general, checking whether there exists an envy-free and Pareto-optimal allocation for goods is $\Sigma_2^p$-complete~\citep{KBKZ09a}.

The idea of envy-freeness up to one good (EF1) was implicit in the paper by \citet{LMMS04a}. Today, it has become a well-studied fairness concept in its own right~\citep{Budi11a}. \citet{CKMPS19} further popularized it, showing that a natural modification of the Nash welfare maximizing rule satisfies EF1 and PO for the case of goods. \citet{BMV17a} recently presented a pseudo-polynomial-time algorithm for computing an allocation that is PO and EF1 for goods. A stronger fairness concept, {\em envy freeness up to the least valued good} (EFX), was introduced by \citet{CKMPS19}.

\citet{Aziz16a} noted that the work on multi-agent chore allocation is less developed than that of goods and that results from one may not necessarily carry over to the other. \citet{ARSW17a} considered fair allocation of indivisible chores and showed that there exists a simple polynomial-time algorithm that returns a 2-approximate MMS allocation for chores. \citet{BK20} presented a better approximation algorithm. \citet{CKKK12} studied the efficiency loss in order to achieve several fair allocations in the context of both good and chore divisions. Allocation of a mixture of goods and chores has received recent attention in the context of divisible items~\citep{BMSY16a,BMSY17a}. Here, we focus on indivisible items.

\paragraph{Subsequent Work.}
Our study of a general setting for goods and chores and our formalization of general definitions for fairness concepts (that apply well to hybrid settings) has spurred further work on the topic. 
In his survey \citet{Moul19a} discusses the subtle differences between the treatment of goods and chores. 
\citet{GaMc20a} and \citet{CGMM21a} focus on algorithms for computing competitive equilibrium for goods and chores. 
\citet{AlWa20a} consider variations of concepts and algorithms that we propose. \citet{AzRe20a} consider a stronger concept of group envy-freeness for goods and chores. \citet{AMS20a} focus on our formulation of PROP1 (that is a weakening of proportionality) and propose a polynomial-time algorithm for computing allocations that are PROP1 and Pareto optimal.

\citet{brczi2020envyfree} pointed out that the natural extension of envy-cycle elimination algorithm of \citet{LMMS04a} does not give an EF1 allocation for a mix of goods and chores. 
\citet{BSV20a} provide further insights into the issue and show that restricting the envy-graph to edges involving maximum envy results in an EF1 allocation for doubly-monotonic valuations that are more general than additive utilities.

\section{Our Model and Fairness Concepts}\label{sec:model}
We now define a fair division problem of indivisible items where agents may have both positive and negative utilities. For a natural number $s \in \bbN$, we write $[s]=\{1,2,\ldots,s\}$. 
An \emph{instance} is a triple $I=(N,O,\calU)$ where
\begin{itemize}
\item $N=[n]$ is a set of {\em agents},
\item $O=\{o_1,o_2,\ldots,o_m\}$ is a set of {\em indivisible items}, and
\item $\calU$ is an $n$-tuple of utility functions $u_i:2^O \rightarrow \bbR$. 
\end{itemize}
Each subset $X \subseteq O$ is referred to as a {\em bundle} of items. We may abuse the notation and write $u_i(o)= u_i(\{o\})$. We note that under this model, an item can be a good for one agent $i$ (i.e., $u_i(o)>0$) but a chore for another agent $j$ (i.e., $u_j(o)<0$). We assume that the utilities in this paper are \emph{additive}, namely, $u_i(X)=\sum_{o \in X}u_i(o)$ for each bundle $X \subseteq O$. We say that agent $i$ \emph{weakly prefers} (respectively, \emph{strictly prefers}) item $o$ to item $o'$ if $u_i(o) \geq u_i(o')$ (respectively, $u_i(o)> u_i(o')$). 
An {\em allocation} $\pi$ is a function $\pi \colon N\rightarrow 2^O$ such that $\bigcup_{i \in N}\pi(i)=O$, and $\pi(i) \cap \pi(j)=\emptyset$ for every pair of distinct agents $i,j \in N$.

We first observe that the definitions of standard fairness concepts can be naturally extended to this general model. The most classical fairness principle is {\em envy-freeness}, requiring that agents do not envy each other. Specifically, given an allocation $\pi$, we say that $i$ {\em envies} $j$ if $u_i(\pi(i)) < u_i(\pi(j))$. An allocation $\pi$ is \emph{envy-free} (EF) if no agent envies the other agents. Another appealing notion of fairness is {\em proportionality}, which guarantees each agent an $1/n$ fraction of her utility for the whole set of items. Formally, an allocation $\pi$ is \emph{proportional} (PROP) if each agent $i \in N$ receives a bundle $\pi(i)$ of utility at least her {\em proportional fair share} $u_{i}(O)/n$. 
The following implication, which is well-known for the case of goods, holds in our setting as well.

\begin{proposition}\label{prop:relation:EF}
For additive utilities, an envy-free allocation satisfies proportionality. 
\end{proposition}
\begin{proof}
Suppose that an allocation $\pi$ is an envy-free allocation. Consider any agent $i \in N$. Then, by envy-freeness, $u_i(\pi(i))\geq u_i(\pi(j))$ for all $j\in N$. Thus, by summing up all the inequalities, $n\cdot u_i(\pi(i))\geq \sum_{j\in N}u_i(\pi(j))=u_{i}(O)$. Hence, each $i\in N$ receives a bundle of utility at least $u_{i}(O)/n$, so $\pi$ satisfies proportionality.
\qed\end{proof}

A simple example of one good with two agents already suggests the impossibility in achieving envy-freeness and proportionality. The recent literature on indivisible allocation has, thereby, focused on approximations of these fairness concepts. A prominent relaxation of envy-freeness, introduced by \citet{Budi11a}, is {\em envy-freeness up to one good} (EF1), which requires that an agent's envy towards another bundle can be eliminated by removing some good from the envied bundle. We will present a generalized definition of EF1 for our setting:
the envy can disappear 
by removing either one ``good" from the other's bundle or one ``chore" from their own bundle. Given an allocation $\pi$, we say that $i$ {\em envies} $j$ \emph{by more than one item} if $i$ envies $j$, and $u_i(\pi(i)\setminus \{o\})  < u_i(\pi(j)\setminus \{o\})$ for any item $o \in \pi(i) \cup \pi(j)$. 
		
\begin{definition}[EF1]
An allocation $\pi$ is \emph{envy-free up to one item (EF1)} if for all $i,j \in N$, $i$ does not envy $j$ by more than one item.
\end{definition}

Obviously, envy-freeness implies EF1. \citet{CFS17a} introduced a novel relaxation of proportionality, which is referred to as {\em PROP1}. In the context of good allocation, this fairness relaxation is a weakening of both EF1 and proportionality, requiring that each agent gets her proportional fair share if she obtains one additional good from the others' bundles. Now we will extend this definition to our setting: under our definition, each agent receives her proportional fair share by obtaining an additional good or removing some chore from her bundle. 
			
\begin{definition}[PROP1]
An allocation $\pi$ satisfies \emph{proportionality up to one item (PROP1)} if for each agent $i\in N$, 
\begin{itemize} 
\item $u_i(\pi(i))\geq u_{i}(O)/n$; or
\item $u_i(\pi(i))+u_i(o)\geq u_{i}(O)/n$ for some $o\in O\setminus \pi(i)$; or 
\item $u_i(\pi(i))-u_i(o)\geq u_{i}(O)/n$ for some $o\in \pi(i)$. 
\end{itemize}
\end{definition}
			
We can verify that EF1 implies PROP1.  
\begin{proposition}\label{prop:relation:EF1Prop1}
For additive utilities, an EF1 allocation satisfies PROP1.
\end{proposition}
\begin{proof}		
The claim clearly holds when $|N| \leq 1$ or $O=\emptyset$; thus suppose $|N| \geq 2$ and $O \neq \emptyset$. Consider any allocation $\pi$ that satisfies EF1, and any agent $i \in N$.

First, consider the case when $\pi(i)=O$. If $u_i(O) \geq 0$, then it is clear that $u_i(\pi(i))=u_i(O) \geq u_i(O)/n$. If $u_i(O) < 0$, consider any $j \in N \setminus \{i\} \neq \emptyset$. Since $\pi$ is EF1, there is an item $o \in \pi(i)$ such that $u_i(\pi(i))- u_i(o) \geq u_i(\pi(j))= u_i(\emptyset)=0$.  Thus, $u_i(\pi(i)) -u_i(o) \geq 0> u_i(O)/n$ for some $o \in \pi(i)$.
    
Next, consider the case when $\pi(i)=\emptyset$. If $u_i(O) \leq 0$, then it is clear that $u_i(\pi(i))=u_i(\emptyset)=0 \geq u_i(O)/n$. Thus, suppose $u_i(O) > 0$. Then, by EF1, for every $j \in N \setminus \{i\}$, $i$ does not envy $j$, or there is an item $o \in \pi(j)$ such that $u_i(\pi(i))= u_i(\emptyset) =0 \geq u_i(\pi(j))- u_i(o)$. Let $o^* \in O$ be such that $o^* \in \argmax_{o \in O}u_i(o)$. Note that $u_i(o^*)>0$ since $u_i(O) > 0$. Then, $u_i(\pi(i)) + u_i(o^*) \geq u_i(\pi(j))$ for every $j  \in N$, which implies that $u_i(\pi(i)) + u_i(o^*) \geq u_i(O)/n$.

Finally, consider the case when when $O\setminus \pi(i) \neq \emptyset$ and $\pi(i) \neq \emptyset$. Let $x=\max_{o\in O\setminus \pi(i)}u_i(o)$ and $y= \min_{o\in \pi(i)}u_i(o)$. Since $\pi$ satisfies EF1, for any agent $j \in N \setminus \{i\}$, 
    \begin{itemize}
        \item $i$ does not envy $j$; or
        \item there exists an item $o \in \pi(j)$ such that $u_i(\pi(i)) \geq u_i(\pi(j)) - u_i(o)$; or 
        \item there exists an item $o \in \pi(i)$ such that $u_i(\pi(i))- u_i(o) \geq u_i(\pi(j))$,
    \end{itemize}
which implies 
\begin{itemize}
    \item $u_i(\pi(i)) \geq u_i(\pi(j))$; or
    \item $u_i(\pi(i)) +x \geq u_i(\pi(j))$; or 
    \item $u_i(\pi(i))- y \geq u_i(\pi(j))$. 
\end{itemize}
Thus, if $i$ gets bonus utility $b^* :=\max \{x,-y,0\}$ by getting some good or removing some chore, her updated utility is such that $u_i(\pi(i)) + b^* \geq u_i(\pi(j))$ for any agent $j \in N \setminus \{i\}$. This would imply that 
$$
n(u_i(\pi(i))+b^*)\geq \sum_{j\in N} u_i(\pi(j))=u_i(O), 
$$
which implies that $u_i(\pi(i))+b^* \geq u_i(O)/n$. Hence PROP1 is satisfied. \qed
\end{proof}

Figure \ref{fig:part-relations} illustrates the relations between fairness concepts introduced above.

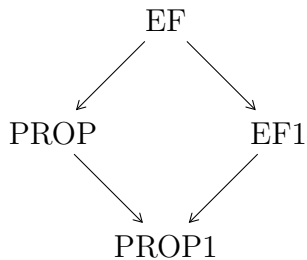
\begin{figure}[h!]
\begin{center}

\scalebox{1}{
\begin{tikzpicture}
\tikzstyle{pfeil}=[->,>=angle 60, shorten >=1pt,draw]
\tikzstyle{onlytext}=[]
\node[onlytext] (EF) at (0,3) {\large EF};
\node[onlytext] (PROP) at (-1.5,1.5) {\large PROP};
\node[onlytext] (PROP1) at (0,0) {\large PROP1};
\node[onlytext] (EF1) at (1.5,1.5) {\large EF1};

\draw[pfeil] (EF) to (EF1);
\draw[pfeil] (EF) to (PROP);
\draw[pfeil] (EF1) to (PROP1);
\draw[pfeil] (PROP) to (PROP1);
\end{tikzpicture}
}
\end{center}	
\caption{Relations between fairness concepts.}
\label{fig:part-relations}
\end{figure}

Besides fairness, we will also consider an efficiency criterion. The most commonly used efficiency concept is {\em Pareto-optimality}. Given an allocation $\pi$, another allocation $\pi'$ is a {\em Pareto-improvement} of $\pi$ if $u_i(\pi'(i)) \geq u_i(\pi(i))$ for all $i \in N$ and $u_j(\pi'(j)) > u_j(\pi(j))$ for some $j \in N$. We say that an allocation $\pi$ is {\em Pareto-optimal} (PO) if there is no allocation that is a Pareto-improvement of $\pi$.		
	
\section{Finding an EF1 Allocation}
In this section, we focus on EF1, a very permissive fairness concept that admits a polynomial-time algorithm in the case of good allocation. For instance, consider a {\em round robin rule} in which agents take turns, and choose their most preferred unallocated item. The round robin rule finds an EF1 allocation if all the items are goods~(see e.g., \citealp{CKMPS19}). By a very similar argument, it can be shown that the algorithm also finds an EF1 allocation if all the items are chores. However, we will show that the round robin rule already fails to find an EF1 allocation if we have some items that are goods and others that are chores. 

\begin{proposition}
The round robin rule does not satisfy EF1.	
\end{proposition}
\begin{proof}
Suppose there are two agents and four items with identical utilities described below. 
\begin{center}
	\upshape
	\setlength{\tabcolsep}{6.4pt}
	\begin{tabular}{rccccc}
		\toprule
		&
		\multicolumn{4}{l}{\!\!\!
			\begin{tikzpicture}[scale=0.57, transform shape, every node/.style={minimum size=7mm, inner sep=1.2pt, font=\huge}]	
			\node[draw, circle](2)  at (1.2,0) {$1$};
			\node[draw, circle](3)  at (2.4,0) {$2$};
			\node[draw, circle](4)  at (3.6,0) {$3$};
			\node[draw, circle](5)  at (4.8,0) {$4$};
			\end{tikzpicture}\!\!\!\!
			\vspace{-2pt}
		} \\
		\midrule
		Alice, Bob:\!\! &  2 & -3 & -3 & -3  \\
		\bottomrule
\end{tabular}
\end{center}
\smallskip
Consider the order, in which Alice chooses the only good and then the remaining chores of equal value are allocated accordingly. In that case, Alice gets the positively valued good and one chore, whereas Bob gets two chores. So even if one item is removed from the bundles of Alice or Bob, Bob will still remain envious. \qed
\end{proof}

Nevertheless, a careful adaptation of the round robin method to our setting, which we call the {\em double round robin algorithm}, constructs an EF1 allocation. In essence, the algorithm will apply the round robin method twice: clockwise and anticlockwise. In the first phase, the round-robin algorithm allocates {\em chores} to agents (i.e., the items for which each agent has non-positive utility), while in the second phase, the reversed round-robin algorithm allocates the remaining {\em goods} to agents, in the opposite order starting with the agent who chooses last in the first phase. Intuitively each agent $i$ may envy agent $j$ who comes earlier than her at the end of one phase, but $i$ does not envy $j$ with respect to the items allocated in the other round; hence the envy of $i$ towards $j$ can be bounded up to one item. We present a formal description of the algorithm in Algorithm~\ref{alg:double}; see Figure \ref{fig:double} for an illustration. 
			
\begin{algorithm}                      
\caption{Double Round Robin Algorithm}         
\label{alg:double}                          
\begin{algorithmic}[1]           
\REQUIRE An instance $I=(N,O,\calU)$.
\ENSURE An allocation $\pi$.
\STATE Initialize $\pi(i)=\emptyset$ for each agent $i \in N$. 
\STATE  Partition $O$ into $O^+=\{o\in O\mid \exists i\in N \text{ s.t. } u_i(o)> 0\}$, $O^-=\{o\in O\mid \forall i\in N,  u_i(o) \leq 0\}$ and suppose $|O^-|=an-k$ for some positive integer $a$ and $k\in \{0,,\ldots, n-1\}$. \label{step:partition}
\STATE Create $k$ dummy chores for which each agent has utility $0$, and add them to $O^-$ (hence, $|O^-|=an$).
\STATE Let the agents come in a round robin sequence $(1,2,\ldots, n)^{*}$ and pick their most preferred item in $O^-$ until all items in $O^-$ are allocated.\label{step:chores}
\STATE Let the agents come in a round robin sequence $(n, n-1,\ldots,1)^{*}$ and pick their most preferred item in $O^+$ until all items in $O^+$ are allocated. If an agent has no available item which gives her strictly positive utility, she does not get a real item but pretends to \emph{pick} a dummy one for which she has utility $0$.\label{step:goods}
\STATE Remove the dummy items from the current allocation $\pi$ and return the resulting allocation $\pi^*$.
\end{algorithmic}
\end{algorithm}

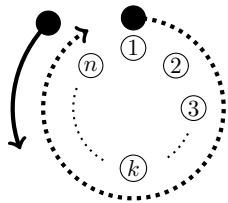
\begin{figure}[htb]
\centering
\begin{tikzpicture}[scale=1.4, transform shape,every node/.style={minimum size=4mm, inner sep=1pt}]
	\def \radius {1cm}
	\node[draw,circle](1) at ({90}:\radius) {$1$};
	\node[draw,circle](2) at ({45}:\radius) {$2$};
	\node[draw,circle](3) at ({0}:\radius) {$3$};
	\draw [dotted,thick] ({-25}:\radius) arc [radius=1, start angle=-25, end angle= -60];
	\draw [dotted,thick] ({240}:\radius) arc [radius=1, start angle=240, end angle= 160];
	
	\node[draw,circle](4) at ({270}:\radius) {$k$};
	\node[draw,circle](6) at ({135}:\radius) {$n$};

         \node[draw,circle,fill=black](A) at ({90}:1.5) {};
	\draw[->,ultra thick,dotted] (A) arc [radius=1.5, start angle=90, end angle= -240];
        
         \node[draw,circle,fill=black](B) at ({135}:2) {};
         \draw[->,ultra thick] (B) arc [radius=2, start angle=135, end angle= 200];
\end{tikzpicture}
\caption{Illustration of Double Round Robin Algorithm. The dotted line corresponds to the picking order when allocating chores. The thick line corresponds to the picking order when allocating goods. The solid black circle indicates the agent who starts the picking. For the dotted (chores) round, agent $1$ is the first agent to pick. For the solid (goods) round, agent $n$ is the first agent to pick.}\label{fig:double}
\end{figure}

In the following, for an allocation $\pi$ and a bundle $X$, we say that $i$ {\em envies} $j$ {\em with respect to} $X$ if $u_i(\pi(i) \cap X) < u_i(\pi(j) \cap X)$. 

\begin{theorem}
For additive utilities, the double round robin algorithm returns an EF1 allocation in $O(\max \{m\log m,mn\})$ time.
\end{theorem}
\begin{proof}
We note that the algorithm ensures that all agents receive the same number of chores, by introducing $k$ dummy chores. Now let $\pi$ be the output of Algorithm \ref{alg:double}. To see that $\pi$ satisfies EF1, consider any pair of two agents $i$ and $j$ where $i<j$. We will show that by removing one item, these agents do not envy each other. We denote by $c^{i}_t$ and $c^{j}_t$ the $t$-th items allocated to agent $i$ and agent $j$ for $t=1,2,\ldots,a$ in Line \ref{step:chores}, respectively. We denote by $g^{i}_t$ and $g^{j}_t$ the $t$-th items allocated to agent $i$ and agent $j$ for $t=1,2,\ldots,b$ in Line \ref{step:goods}, respectively, where $b$ denotes the number of rounds each agent chooses an item (including a dummy item) in Line \ref{step:goods}. 

First, consider $i$'s envy for $j$. We first observe that the $t$-th item $c^{i}_t$ in $O^-$ allocated to $i$ is weakly preferred by $i$ to the $t$-th item $c^{j}_t$ in $O^-$ allocated to $j$. Hence, agent $i$ does not envy $j$ with respect to $O^-$. Namely, 
\begin{align}
&u_i(\pi(i) \cap O^-)= \sum^a_{t=1}u_i(c^{i}_t) \geq \sum^a_{t=1}u_i(c^{j}_t) =u_i(\pi(j) \cap O^-). \label{eq:chore}
\end{align}
As for the good allocation, agent $i$ may envy agent $j$ with respect to $O^+$. But if the first item $g^{j}_1$ picked by $j$ from $O^+$ is removed from $j$'s bundle, then the envy will
disappear,
i.e., $i$ does not envy $j$ with respect to $O^+\setminus \{g^{j}_1\}$. Namely, 
\begin{align}
&u_i(\pi(i) \cap O^+)= \sum^b_{t=1}u_i(g^{i}_t) \geq \sum^b_{t=2}u_i(g^{j}_t) = u_i((\pi(j) \cap O^+)\setminus \{g^{j}_1\}). \label{eq:good} 
\end{align}
The reason is that for each item $g^{j}_t$ picked by $j$ where $t=2,3,\ldots,b$, there is a corresponding item $g^{i}_{t-1}$ picked by $i$ before $j$'s turn that is weakly as preferred by $i$ to $g^{j}_t$. Combining \eqref{eq:chore} and \eqref{eq:good} yields $u_i(\pi(i)) \geq u_i(\pi(j)\setminus \{g^{j}_1\})$.

Second, consider $j$'s envy for $i$. Similarly to the above, agent $j$ does not envy agent $i$ with respect to $O^+$ because agent $j$ takes the first pick among $i$ and $j$; that is, for every item in $g^{i}_t$ chosen by $i$, agent $j$ picks an item $g^{j}_t$ before $i$ that she weakly prefers to $g^{i}_t$. Thus, 
\begin{align}
&u_j(\pi(j) \cap O^+)= \sum^b_{t=1}u_j(g^{j}_t) \geq \sum^b_{t=1}u_i(g^{i}_t) =u_j(\pi(i) \cap O^+). \label{eq:good:ji} 
\end{align}
As for the items in $O^-$, for each item $c^{i}_{t}$ picked by $i$ where $t=2,3,\ldots,a$, there is an item $c^{j}_{t-1}$ picked by $j$ before $i$ that $j$ weakly prefers to $c^{i}_{t}$, which implies that $j$ does not envy $i$ with respect to $O^-\setminus \{c^{j}_a\}$. Thus 
\begin{align}
&u_j(\pi(j)\setminus \{c^{j}_a\})= \sum^{a-1}_{t=1}u_j(c^{j}_t) \geq \sum^a_{t=2}u_j(c^{i}_t) \geq \sum^a_{t=1}u_j(c^{i}_t) = u_j(\pi(i)). \label{eq:chore:ji} 
\end{align}
Note that the last inequality holds since $u_j(c^{i}_1) \leq 0$. Combining \eqref{eq:good:ji} and \eqref{eq:chore:ji} yields $u_j(\pi(j)\setminus \{c^{j}_a\}) \geq u_j(\pi(i))$.

In either case, agents do not envy each other by more than one item. We conclude that $\pi$ is EF1 and so does the final allocation $\pi^*$ as removing dummy items does not affect the utilities of each agent. 

It remains to analyze the running time of Algorithm \ref{alg:double}. Line \ref{step:partition} requires $O(mn)$ time as each item needs to be examined by all agents. Lines \ref{step:chores} and \ref{step:goods} require $O(m\log m)$ time as there are at most $m$ iterations, and for each iteration, each agent has to choose the most preferred item out of at most $m$ items, which can be done by sorting all the items according to the preference of each agent at the beginning. Thus, the total running time can be bounded by $O(\max \{m\log m,mn\})$, which completes the proof.\qed
\end{proof}

\section{Finding an EF1 and PO allocation}
We move on to the next question as to whether fairness is achievable together with efficiency. In the context of good allocation where agents have non-negative additive utilities, \citet{CKMPS19} proved that an outcome that maximizes the {\em Nash welfare} (i.e., the product of utilities) satisfies EF1 and Pareto-optimality simultaneously. The question regarding whether a Pareto-optimal and EF1 allocation exists for chores is unresolved. Starting from an EF1 allocation and finding Pareto improvements, one runs into two challenges: first, Pareto improvements may not necessarily preserve EF1; second, finding Pareto improvements is NP-hard~\citep{ABL+16a,KBKZ09a}. Even if we ignore the second challenge, the question regarding the existence of a Pareto-optimal and EF1 allocation for chores is open. Next, we show that the problem of finding an EF1 and Pareto-optimal allocation is completely resolved for the restricted but important case of two agents. The main theorem in this section is stated as follows. 

\begin{theorem}\label{thm:AW}
For two agents with additive utilities, a Pareto-optimal and EF1 allocation always exists and can be computed in $O(m^2)$ time.
\end{theorem}

Our algorithm for the problem can be viewed as a discrete version of the well-known Adjusted Winner (AW) rule~\citep{BT96a,BT96b}. Just like the Adjusted Winner rule, our algorithm finds a Pareto-optimal and EF1 allocation. In contrast to AW, which is designed for goods, our algorithm can handle both goods and chores. 

The algorithm begins by giving each subjective item to the agent who considers it as a good. 
So, in the following, we assume that we have objective items only, i.e., for each item $o \in O$, either $o$ is a good ($u_i(o)>0$ for each $i \in N$); or $o$ is a chore ($u_i(o)<0$ for each $i \in N$).
Now we call one of the two agents {\em winner} (denoted by $w$) and another {\em loser} (denoted by $\ell$). 
\begin{enumerate}
\item Initially, all goods are allocated to the winner and all chores to the loser. 
\item We sort the items in terms of $|u_{\ell}(o)|/|u_w(o)|$ (monotone non-increasing order), and consider reallocation of the items according to the ordering (from the left-most to the right-most item). 
\item When considering a good, we move it from the winner to the loser. When considering a chore, we move it from the loser to the winner. We stop when the loser does not envy the winner by more than one item. 
\end{enumerate}
We present a formal description of the algorithm in Algorithm \ref{alg:AW}.

\begin{algorithm}                      
\caption{Generalized Adjusted Winner Algorithm}         
\label{alg:AW}                          
\begin{algorithmic}[1]           
\REQUIRE An instance $I=(N,O,\calU)$ where $N=\{w,\ell \}$.
\ENSURE An allocation $\pi$.
\STATE Initialize $\pi(i)=\emptyset$ for each agent $i \in N$. 
\STATE Let $O^*_{w}=\{\, o \in O \mid u_{w}(o)\geq 0~\mbox{and}~u_{\ell}(o) \leq 0 \,\}$ and $O^*_{\ell}=\{\, o \in O \mid u_{\ell}(o)\geq 0~\mbox{and}~u_{w}(o) < 0  \,\}$. 
\STATE Let $O^+=\{\, o \in O  \mid u_i(o) > 0\quad \forall i \in N\,\}$ and 
$O^-=\{\,o \in O \mid u_i(o) <0\quad  \forall i \in N\,\}$. 
\STATE For each item $o \in O^+ \cup O^*_w$, allocate $o$ to agent $w$. For each item $o \in O^{-} \cup O^*_{\ell}$, allocate $o$ to agent $\ell$.\label{line:initialization}
\STATE Sort the items in $O^+ \cup O^-=\{o_1,o_2,\ldots,o_r \}$ where 
$|u_{\ell}(o_1)|/|u_w(o_1)| \ge |u_{\ell}(o_2)|/|u_w(o_2)| \ge \cdots \ge |u_{\ell}(o_r)|/|u_w(o_r)|$. 
\STATE Set $t=1$. 
\WHILE{agent $\ell$ envies agent $w$ by more than one item} 
\IF{$o_t \in O^+$}\label{line:while}
\STATE Set $\pi(w)=\pi(w) \setminus \{o_t\}$ and $\pi(\ell)=\pi(\ell) \cup \{o_t\}$. 
\ELSIF{$o_t \in O^-$}
\STATE Set $\pi(w)=\pi(w) \cup \{o_t\}$ and $\pi(\ell)=\pi(\ell) \setminus \{o_t\}$. 
\ENDIF
\STATE Update $t = t +1$. 
\ENDWHILE 
\end{algorithmic}
\end{algorithm}

We will first prove that at any point of the algorithm, the allocation $\pi$ is Pareto-optimal, and so is the final allocation. 
\begin{lemma}\label{lem:PO:AW}
During the execution of Algorithm \ref{alg:AW}, the allocation $\pi$ is Pareto-optimal at any point after Line \ref{line:initialization}. 
\end{lemma}
\begin{proof}
It can be easily verified that the allocation $\pi$ just after Line \ref{line:initialization} is Pareto-optimal. 
Thus, consider some time step after the algorithm enters the {\bf while}-loop of Line \ref{line:while}. Assume towards a contradiction that $\pi'$ is a Pareto-improvement of $\pi$. We assume without loss of generality that all items in $O^*_w$ remain assigned to $w$ under $\pi'$ because transferring an item in $O^*_w$ from $w$ to $\ell$ improve neither the utility of $w$ nor that of $\ell$. Likewise, we assume that all items in $O^*_{\ell}$ remain assigned to $\ell$ under $\pi'$.

In the following, we call each item $o \in O^+$ a \emph{good} and item $o \in O^-$ a \emph{chore}. 
For each $i,j \in \{ w,\ell\}$ with $i \neq j$, let
\begin{itemize}
\item $G_{ii}$ be the set of goods in $\pi(i) \cap \pi'(i)$;
\item $C_{ii}$ be the set of chores in $\pi(i) \cap \pi'(i)$;
\item $G_{ij}$ be the set of goods in $\pi(i) \cap \pi'(j)$;
\item $C_{ij}$ be the set of chores in $\pi(i) \cap \pi'(j)$.
\end{itemize}

Consider first the case when in $\pi$, the winner has utility which is at least as high as in $\pi'$, while the loser is strictly better off. Taking into account that the bundles of goods $G_{ww}$ and $G_{\ell \ell}$ and the bundles of chores $C_{ww}$ and $C_{\ell \ell}$ are allocated to the same agent in both allocations, this means
\begin{align}\label{eq:pareto-dom1}
u_{w}(G_{\ell w})+u_{w}(C_{\ell w})-u_{w}(G_{w \ell})-u_{w}(C_{w \ell}) &\geq 0;~\mbox{and}\\\label{eq:pareto-dom2}
u_{\ell}(G_{w \ell})+u_{\ell}(C_{w \ell})-u_{\ell}(G_{\ell w})-u_{\ell}(C_{\ell w}) &>0
\end{align}

The crucial observation now is that the algorithm considered all items in $G_{\ell w}$ and $C_{w \ell}$ before the items in $G_{w \ell}$ and $C_{\ell w}$ in the ordering. 
Indeed, recall that all the goods are initially assigned to the winner, $G_{\ell w} \subseteq \pi(\ell)$, and $G_{w \ell} \subseteq \pi(w)$. Thus the goods in $G_{\ell w}$ are those transferred from the winner $w$ to the loser $\ell$ in the {\bf while}-loop of Line \ref{line:while}, while the goods in $G_{w \ell}$ are those that stay in the winner's bundle. Similarly, recall that all the chores are initially assigned to the loser, $C_{w \ell} \subseteq \pi(w)$, and $C_{\ell w} \subseteq \pi(\ell)$. Thus, the chores in $C_{w \ell}$ are those transferred from the loser $\ell$ to the winner $w$, while the chores in $C_{\ell w}$ are those that stay in the loser's bundle.
Now, let $\alpha$ be such that
\begin{align*}
\max_{o\in G_{w \ell}\cup C_{\ell w}}{|u_{\ell}(o)|/|u_{w}(o)|} \leq \alpha \leq \min_{o\in G_{\ell w}\cup C_{w \ell}}{|u_{\ell}(o)|/|u_{w}(o)|}.
\end{align*}
This definition implies the inequalities, 
\begin{align*}
u_{\ell}(G_{w \ell})\leq \alpha u_{w}(G_{w \ell}); 
u_{\ell}(G_{\ell w})\geq \alpha u_{w}(G_{\ell w}); \\
-u_{\ell}(C_{w \ell})\geq -\alpha u_{w}(C_{w \ell});
-u_{\ell}(C_{\ell w})\leq -\alpha u_{w}(C_{\ell w}),
\end{align*}
which, together with inequality (\ref{eq:pareto-dom2}), yield
\begin{align*}
0 &< u_{\ell}(G_{w \ell}) + u_{\ell}(C_{w \ell})- u_{\ell}(G_{\ell w}) - u_{\ell}(C_{\ell w}) \\
&\leq-\alpha(u_{w}(G_{\ell w})+u_{w}(C_{\ell w})-u_{w}(G_{w \ell}) - u_{w}(C_{w \ell}))\leq 0,
\end{align*}
a contradiction. The last inequality follows by (\ref{eq:pareto-dom1}) and by the fact that $\alpha$ is non-negative. A similar argument applies when in $\pi$, the loser has utility which is at least as high as in $\pi'$, while the winner is strictly better off. \qed
\end{proof}

We are now ready to prove Theorem \ref{thm:AW}. 

\begin{proof}[of Theorem \ref{thm:AW}]
We will prove that the final output $\pi$ of Algorithm \ref{alg:AW} satisfies EF1. Together with Lemma \ref{lem:PO:AW}, this proves the desired claim. 
Now observe that at the final allocation $\pi$, at most one agent envies the other: if the loser still envies the winner and the winner also envies the loser, then exchanging the bundles would result in a Pareto improvement, contradicting Lemma~\ref{lem:PO:AW}. Thus, $\pi$ is EF1 when the loser envies the winner at $\pi$. Consider when at $\pi$, the loser does not envy the winner but the winner envies the loser. Let $\pi'$ be the previous allocation just before the final transfer in the {\bf while}-loop of Line \ref{line:while}. Let $W=\pi'(w) \cap \pi(w)$ and $L=\pi'(\ell) \cap \pi(\ell)$. Namely, $W$ (respectively, $L$) is the set of items in the winner's bundle (respectively, the loser's bundle) excluding the transferred item at $\pi'$ and $\pi$. By construction, the loser envies the winner by more than one item at $\pi'$, which implies $u_{\ell}(L) < u_{\ell}(W)$. Suppose towards a contradiction that the winner envies the loser by more than one item at $\pi$, which implies $u_{w}(W) < u_{w}(L)$.

\begin{itemize}
    \item If $g$ is the last good that has been moved from the winner to the loser, then allocating $W$ to $\ell$ and $L \cup \{g\}$ to $w$ would be a Pareto-improvement of $\pi'$, a contradiction. 
    \item If $c$ is the last chore that has been moved from the loser to the winner, then allocating $W\cup \{c\}$ to $\ell$ and $L$ to $w$ would be a Pareto-improvement of $\pi'$, a contradiction.
\end{itemize}
Hence, the winner does not envy the loser by more than one item at $\pi$; we conclude that $\pi$ is EF1.

It remains to analyze the running time of the algorithm. First, the items can be sorted in $O(m \log m)$ time. The adjustment process takes $O(m^2)$ time. Each iteration checks the allocation is EF1 from the view point of the loser, which requires at most $m$ comparisons of utilities, and there are at most $m$ iterations. Thus, the number of operations is bounded by $O(m^2)$.\qed
\end{proof}

The example below illustrates our discrete adaptation of AW.
\begin{example}[Example of the generalized AW]
Consider two agents, Alice and Bob, and five items with the following additive utilities where the gray circles correspond to goods and the white circles correspond to chores. 
\begin{center}
	\upshape
	\setlength{\tabcolsep}{6.4pt}
	\begin{tabular}{rcccccccc}
		\toprule
		&
		\multicolumn{7}{l}{\!\!\!
			\begin{tikzpicture}[scale=0.57, transform shape, every node/.style={minimum size=7mm, inner sep=1.3pt, font=\huge}]	
			\node[draw, circle, fill=gray!70](1)  at (1.2,0) {$1$};
			\node[draw, circle](2)  at (2.4,0) {$2$};
			\node[draw, circle,fill=gray!70](3)  at (3.6,0) {$3$};
			\node[draw, circle, fill=gray!70](4)  at (4.8,0) {$4$};
			\node[draw, circle](5)  at (6,0) {$5$};
			\node[draw, circle](6)  at (7.3,0) {$6$};
			\node[draw, circle](7)  at (8.6,0) {$7$};
			\end{tikzpicture}\!\!\!\!
			\vspace{-2pt}
		} \\
		\midrule
		Alice (winner) :\!\! & 1 & -1 & 2 & 1 &-2&-4& -6 \\
		Bob (loser) :\!\! &  4 & -3 & 6 & 2 & -2 & -2  & -2\\
		$|u_{\ell}(o)|/|u_w(o)|$  :\!\! &  4 & 3 & 3 & 2 & 1 & 1/2  & 1/3\\
		\bottomrule
\end{tabular}
\end{center}
The generalized AW initially allocates the goods to Alice and the chores to Bob. Then, it transfers the first good from Alice to Bob and moves the second chore from Bob to Alice. After moving the third good from Alice to Bob, Bob stops being envious by more than one item. Hence the final allocation gives the items $2$ and $4$ to Alice and the rest to Bob.\qed
\end{example}
			
A natural question is whether PO and EF1 allocations exist for three or more agents; we leave this an an interesting open question. We remark that Pareto-optimality by itself is easy to achieve in $O(nm)$ time. It suffices to give each item to the agent who values it the most.

\section{Finding a Connected PROP1 Allocation}
We saw that there always exists an EF1 allocation for subjective goods and chores. If we weaken EF1 to PROP1, one can achieve one additional requirement besides fairness, that is, {\em connectivity}. In this section, we will consider a situation when items are placed on a path, and each agent desires a connected bundle of the path. Finding a connected set of items is important in many scenarios. For example, the items can be a set of rooms in a corridor and the agents could be research groups where each research group wants to get adjacent rooms (see e.g., \citealp{BCE17a,BCFIMPVZ19}). 

We will show that a connected PROP1 allocation exists and can be found efficiently. In what follows, we assume that the path is given by a sequence of items $(o_1,o_2,\ldots,o_m)$. Formally, we say that an allocation $\pi$ is \emph{connected} if for each agent $i \in N$, $\pi(i)$ is connected in the path $(o_1,o_2,\ldots,o_m)$. We will consider a slightly more stringent notion of PROP1: A connected allocation $\pi$ is PROP1$_{outer}$ if for each agent $i \in N$, 
\begin{itemize} 
\item agent $i$ receives a bundle of utility at least her proportional fair share, i.e., $u_i(\pi(i))\geq u_{i}(O)/n$, or
\item $u_i(\pi(i))+u_i(o)\geq u_{i}(O)/n$ for some item $o\in O\setminus \pi(i)$ such that $\pi(i) \cup \{o\}$ is connected; or 
\item $u_i(\pi(i))-u_i(o)\geq u_{i}(O)/n$ for some $o\in \pi(i)$ such that $\pi(i) \setminus \{o\}$ is connected. 
\end{itemize}

We first prove a result for a case of the cake cutting setting that is of independent interest. In the following, a {\em mixed cake} is the interval $[0,m]$. Each agent $i \in N$ has a value density function ${\hat u}_i$, which maps a subinterval of the cake to a real value, where $i$ has uniform utility $u_i(o_j)$ for the interval $[j-1,j]$ for each $j \in [m]$. The \emph{proportional fair share} of agent $i$ for a mixed cake $[0,m]$ is given by ${\hat u}_i([0,m])/n$.
A {\em contiguous allocation} of a mixed cake assigns each agent a disjoint sub-interval of the cake where the union of the intervals equals the entire cake $[0,m]$; it satisfies {\em proportionality} if each agent $i$ gets an interval of utility at least her proportional fair share.

\begin{theorem}\label{thm:proportional}
For additive utilities, a contiguous proportional allocation of a mixed cake exists and can be computed in polynomial time.
\end{theorem}
\begin{proof}
Let $N^+$ be the set of agents with strictly positive total utility for $O$. 

We combine the moving-knife algorithms for goods and chores as follows. First, if there is an agent who has positive proportional fair share, i.e., $N^+ \neq \emptyset$, we apply the moving-knife algorithm only to the agents in $N^+$. Our algorithm moves a knife from \emph{left} to \emph{right}, and agents shout whenever the left part of the cake has a utility of exactly equal to the proportional fair share. The first agent who shouts is allocated the left bundle, and the algorithm recurs on the remaining instance. Second, if no agent has a positive proportional fair share, our algorithm moves a knife from \emph{right} to \emph{left}, and agents shout whenever the left part of the cake has utility exactly proportional fair share. Again, the first agent who shouts is allocated the left bundle, and the algorithm recurs on the remaining instance. 

\begin{algorithm}                      
\caption{Generalized Moving-knife Algorithm $\calA$}         
\label{alg:movingknife}                          
\begin{algorithmic}[1]  
\REQUIRE A sub-interval $[\ell,r]$, agent set $N'$, utility functions ${\hat u}_i$ for each $i \in N'$.
\ENSURE An allocation ${\hat \pi}$ of a mixed cake $[\ell,r]$ to $N'$.
\STATE Initialize ${\hat \pi}(i)=\emptyset$ for each $i \in N'$.
\STATE Set $N^+=\{\, i \in N' \mid {\hat u}_i([\ell,r]) > 0\,\}$.
\IF{$N^+ \neq \emptyset$}
\IF{$|N^+|=1$}
\STATE Allocate $[\ell,r]$ to the unique agent in $N^+$.
\ELSE
\STATE Let $x_i$ be the minimum point where ${\hat u}_i([\ell,x_i])={\hat u}_i([\ell,r])/|N^+|$ for $i \in N^+$. 
\STATE Find agent $j$ with minimum $x_j$ among all agents in $N^+$. 
\RETURN ${\hat \pi}$ where ${\hat \pi}(j)=[\ell,x_{j}]$ and ${\hat \pi}|_{N'\setminus \{j\}}=\calA([x_{j},r],N' \setminus \{j\},({\hat u}_i)_{i \in N'\setminus \{j\}})$
\ENDIF
\ELSE 
\STATE Let $x_i$ be the maximum point where ${\hat u}_i([\ell,x_i])=-{\hat u}_i([\ell,r])/n$ for $i \in N'$. 
\STATE Find agent $j$ with maximum $x_j$ among all agents in $N'$.
\RETURN ${\hat \pi}$ where ${\hat \pi}(j)=[\ell,x_{j}]$ and ${\hat \pi}|_{N'\setminus \{j\}}=\calA([x_{j},r],N' \setminus \{j\},({\hat u}_i)_{i \in N'\setminus \{j\}})$
\ENDIF
\end{algorithmic}
\end{algorithm}

Algorithm \ref{alg:movingknife} formalizes the idea. To prove its correctness, we will prove by induction on the number of agents $|N'|$ that the allocation of a mixed cake $[\ell,r]$ produced by $\calA$ satisfies the following:
\begin{itemize}
\item if $N^+ \neq \emptyset$, then each agent in $N^+$ receives an interval of utility at least her proportional fair share ${\hat u}_i([\ell,r])/|N'|$ and each agent not in $N^+$ receives an empty piece; and
\item if $N^+ = \emptyset$, then each agent receives an interval of utility at least her proportional fair share ${\hat u}_i([\ell,r])/|N'|$. 
\end{itemize}
The claim is clearly true when $|N'|=1$. Suppose that $\calA$ returns a proportional allocation of a mixed cake with desired properties when $|N'| = k-1$; we will prove it for $|N'|=k$. 

Suppose that some agent has positive proportional fair share, i.e., $N^+ \neq \emptyset$. Note that each agent $i$ not in $N^+$ has non-positive proportional fair share and gets nothing; thus it suffices to show that the agents in $N^+$ receive bundles of utility at least her proportional fair share. If $|N^+|=1$, the claim is trivial; thus assume otherwise. Clearly, agent $j$ receives an interval of utility at least her proportional fair share. Further, all other agents in $N^+$ have utility at most their proportional fair shares for the left piece $[\ell,x_j]$. Indeed, if there is an agent $i' \in N^+$ whose utility for the left piece $[\ell,x_j]$ is greater than her proportional fair share ${\hat u}_{i'}([\ell,r])/k$, then $i'$ would have shouted when the knife reaches before $x_j$ by the continuity of ${\hat u}_{i'}$, i.e., $x_{i'}<x_j$, contradicting the minimality of $x_j$. Thus, the remaining agents in $N^+$ have at least $(k-1) \cdot {\hat u}_i([\ell,r])/k$ utility for the rest of the cake $[x_j,r]$; hence, by the induction hypothesis each agent in $N^+$ gets an interval of utility at least her proportional fair share, and each of the remaining agents gets an empty piece.

Suppose that no agent has positive proportional fair share. Again, if there is an agent $i'$ whose utility for the left piece $[\ell,x_j]$ is greater than her proportional fair share ${\hat u}_{i'}([\ell,r])/k$, then $i'$ would have shouted when the knife reaches before $x_j$ by the continuity of ${\hat u}_{i'}$, i.e., $x_{i'}>x_j$, contradicting the maximality of $x_j$. Thus, all the remaining agents have utility at least $(k-1) \cdot {\hat u}_i([\ell,r])/k$ for the rest of the cake $[x_j,r]$, and hence, by the induction hypothesis, each agent gets an interval of utility at least her proportional fair share ${\hat u}_i([\ell,r])/k$. 
It can be easily verified that Algorithm \ref{alg:movingknife} runs in polynomial time. 
\qed
\end{proof}

The theorem stated above also applies to a general cake-cutting model in which information about agent's utility function over an interval can be inferred by a series of queries. We note that in contrast with proportionality, the existence of a contiguous envy-free allocation of a mixed cake remains elusive: it is known to exist only when the number $n$ of agents is four or a prime number \citep{Halevi2018,MeSh18a}. Next, we show how a fractional proportional allocation can be used to achieve a contiguous PROP1 division of indivisible items. 	

			 
\begin{theorem}
For additive utilities, a connected PROP1$_{outer}$ allocation of a path always exists and can be computed in polynomial time. 
\end{theorem}
\begin{proof}
Given a path $(o_1,o_2,\ldots,o_m)$, consider a mixed cake $[0,m]$ and each agent with a value density function ${\hat u}_i$, where $i$ has uniform utility $u_i(o_j)$ for the interval $[j-1,j]$ for each $j \in [m]$. We know that this instance admits a contiguous and proportional allocation ${\hat \pi}$ from Theorem \ref{thm:proportional}. 
Suppose without loss of generality that under such an allocation ${\hat \pi}$, agents $1, 2,\ldots, n$ receive the 1st, 2nd, $\ldots$, and $n$-th bundles from left to right. That is, each agent $i=1,2,\ldots,n$ receives the sub-interval $[x_{i-1},x_i]$ of the mixed cake, where $0=x_0 \leq x_1  \leq  \ldots  \leq x_{n-1} \leq x_n=m$. Without loss of generality, we also assume that no agent gets the empty bundle under this fractional allocation, i.e., $x_{i-1}< x_i$ for each $i=1,2,\ldots, n$.

From left to right, we show how to allocate each item $o_j$ for $j=1,2,\ldots,m$ to construct an integral allocation ${\pi}$. If item $o_j$ is fully contained in some agent's bundle, namely, $x_{i-1} \leq  j-1   \leq j \leq x_i$ for some $i \in N$, then we assign each item $o_j$ to agent $i$. If not, i.e., the item $o_j$ is on the boundary, we allocate it according to the left-most/right-most agents' preferences. Formally, suppose that $j-1 \leq x_{\ell} \leq x_{\ell+1} \leq \ldots \leq x_r \leq j$ such that $x_{\ell}=\min \{\, x_i \mid x_i \geq j-1 \,\}$ and $x_r=\max \{\, x_i \mid x_i \leq j \,\}$. Then we do the following: 
\begin{enumerate}
\item If two agents $\ell$ and $r$ disagree on the sign of $o_j$, i.e., $\min \{u_{\ell}(o_j), u_{r}(o_j)\} <0 < \max \{u_{\ell}(o_j), u_{r}(o_j)\}$, we give the item $o_j$ to the agent $i \in \{\ell,r\}$ who has positive utility for it.
\item If two agents $\ell$ and $r$ agree on the sign of $o_j$, i.e., $\min \{u_{\ell}(o_j), u_{r}(o_j)\} \geq 0$ or $\max \{u_{\ell}(o_j), u_{r}(o_j)\} < 0$, we allocate the item $o_j$ in such a way that: 
\begin{itemize}
\item the left-agent $\ell$ takes $o_j$ if both agents have non-negative utility, i.e., $\min \{u_{\ell}(o_j), u_{r}(o_j)\} \geq 0$;
\item the right-agent $r$ takes $o_j$ if both agents have negative utility, i.e., $\max \{u_{\ell}(o_j), u_{r}(o_j)\} < 0$.
\end{itemize}
\end{enumerate}
Note that if there is an agent who gets a fraction of one item only under the proportional fractional division, the agent gets nothing under our final allocation. 

The resulting integral allocation $\pi$ is PROP1$_{outer}$. To see this, take any agent $i$. Clearly, when one of the knife positions $x_{i-1}$ and $x_i$ is integral, the bundle satisfies PROP$_{outer}$. Further, if $[x_{i-1},x_j] \subseteq [j-1,j]$ for some $j \in [m]$, agent $i$ gets utility $1/n$ by receiving either the item $o_j$ or the empty bundle.
Thus, assume otherwise, i.e., $x_{i-1},x_i \not \in \{0,1,\ldots,m\}$ and $|x_i-x_{i-1}|>1$. We will show that such an agent gets utility $1/n$ by either receiving the item next to its bundle or deleting the left-most item of her bundle. Let $o_r$ and $o_{\ell}$ be the left and right boundary items where $x_{i-1} \in (r-1,r)$ and $x_{i} \in (\ell-1,\ell)$. Note that we have $\{o_{\ell+1},o_{\ell+2},\ldots,o_{r-1}\} \subseteq \pi(i)$. 
Consider the following four cases. 
\begin{itemize}
    \item Both $o_{\ell}$ and $o_r$ are goods for $i$, i.e., $\min \{u_i(o_{\ell}),u_i(o_r)\} \geq 0$. In this case, agent $i$ receives at least $o_r$. Thus, if $o_{\ell} \in \pi(i)$, agent $i$ obtains utility $1/n$. If not, agent $i$ gets utility $1/n$ by receiving the item $o_{\ell}$. 
    \item Both $o_{\ell}$ and $o_r$ are chores for $i$, i.e., $\max \{u_i(o_{\ell}),u_i(o_r)\} < 0$. In this case, agent $i$ does not receive $o_r$. Thus, if $o_{\ell} \not \in \pi(i)$, agent $i$ obtains utility $1/n$. If not, agent $i$ gets utility $1/n$ by removing the item $o_{\ell}$. 
    \item The item $o_{\ell}$ is a good but $o_{r}$ is a chore for $i$, i.e., $u_i(o_{\ell}) \geq 0$ and $u_i(o_r) <0$. In this case, agent $i$ does not receive $o_r$. Thus, if $o_{\ell} \in \pi(i)$, agent $i$ obtains utility $1/n$. If not, agent $i$ gets utility $1/n$ by receiving the item $o_{\ell}$.
    \item The item $o_{\ell}$ is a chore but $o_{r}$ is a good for $i$, i.e., $u_i(o_{\ell}) < 0$ and $u_i(o_r) \geq 0$. In this case, agent $i$ receives at least $o_r$. Thus, if $o_{\ell} \not \in \pi(i)$, agent $i$ obtains utility $1/n$. If not, agent $i$ gets utility $1/n$ by removing the item $o_{\ell}$.
\end{itemize}
We conclude that $\pi$ is a connected PROP1$_{outer}$ allocation. By Theorem \ref{thm:proportional}, it is immediate to see that one can compute a connected PROP1$_{outer}$ allocation in polynomial time.\qed
\end{proof}
						 
\section{Discussion}
In this paper, we have formally studied fair allocation when the items are a combination of subjective goods and chores. 
Our work paves the way for detailed examination of allocation of goods/chores, and opens up an interesting line of research, with many problems left open to explore.
Perhaps the most intriguing open question as a result of our study is the existence of an EF1 allocation under arbitrary non-monotonic utilities.

There are further fairness concepts that could be studied from both existence and complexity issues, most notably {\em envy-freeness up to the least valued item} (EFX) \citep{CKMPS19}. In our setting, one can define an allocation $\pi$ to be \emph{EFX} if for any pair of agents $i$ and $j$, agent $i$ does not envy agent $j$, or the following two conditions hold: 

\begin{enumerate}
\item $\forall o\in \pi(i)$ s.t. $u_i(\pi(i)\setminus \{o\})>u_i(\pi(i))$: $u_i(\pi(i)\setminus \{o\})\geq u_i(\pi(j))$; and
\item $\forall o\in \pi(j)$ s.t. 
$u_i(\pi(j)\setminus \{o\})<u_i(\pi(j))$
: $u_i(\pi(i))\geq u_i(\pi(j)\setminus \{o\})$.
\end{enumerate} 

That is, $i$'s envy towards $j$ can be eliminated by either removing $i$'s least valuable good from $j$'s bundle or removing $i$'s favorite chore from $i$'s bundle. 
\citet{CKMPS19} mentioned the following `enigmatic' problem:
does an EFX allocation exist for goods? It would be interesting to investigate the same question for subjective or objective goods/chores under additive utilities.

We also note that while the relationship between Pareto-optimality and most fairness notions is still unclear, \citet{CFS17a} proposed a fairness concept called {\em round robin share} (RRS) that can be achieved together with Pareto-optimality. In our context, RRS can be formalized as follows. Given an instance $I=(N,O, U)$, consider the round robin sequence in which all agents have the same utilities as agent $i$. In that case, the minimum utility achieved by any of the agents is RRS$_i(I)$. 
An allocation satisfies RRS if each agent $i$ gets utility at least RRS$_i(I)$. It would be then very natural to ask what is the computational complexity of finding an allocation satisfying both properties. 

Finally, recent papers of \citet{BCE17a} and Bil\`{o} et al. $(2019)$ 
showed that a connected allocation satisfying several fairness notions, such as MMS and EF1, is guaranteed to exist for some restricted domains. These existence results crucially rely on the fact that the agents have monotonic valuations, and it remains open whether similar results can be obtained in fair division of indivisible goods and chores.  

\section*{Acknowledgements}
We would like to thank \citet{brczi2020envyfree}
and \citet{BSV20a}
for pointing out an error of the IJCAI version \citep{ijcai2019-8}. We thank the IJCAI 2019 reviewers and Warut Suksompong for helpful feedback. Ayumi Igarashi is supported by the KAKENHI Grant-in-Aid for JSPS Fellows number 18J00997. Toby Walsh is funded by the European Research Council under the Horizon 2020 Programme via AMPLify 670077.


%
%

\bibliographystyle{plainnat}

\end{document}